\documentclass[journal]{IEEEtran}

\usepackage{verbatim, xspace, algorithm,units, graphicx,amsmath,amssymb,url,array, multirow,caption,color, subfig,algpseudocode, setspace, bbm, ntheorem,unicode,makecell}
\usepackage{booktabs}
\usepackage[]{units}
\usepackage[normalem]{ulem}
\usepackage{aecompl}
\usepackage{bm}
\newtheorem{proof}{Proof}
\newcommand{\etal}{et al.\@\xspace}

\newtheorem{theorem}{Theorem}

\begin{document}

\title{A Scalable Algorithm for Privacy-Preserving Item-based Top-N Recommendation}

\author{Yingying Zhao$^\dag$, Dongsheng Li$^\S$, Qin Lv$^\ddag$, Li Shang$^{\dag, \ddag}$\\
    $^\dag$ Tongji University, Shanghai 201804 P. R. China\\
   $^\S$ IBM Research - China, Shanghai 201203 P. R. China \\
 $^\ddag$ University of Colorado Boulder, Boulder, CO 80309 USA\\
 }
\maketitle

\begin{abstract}
Recommender systems have become an indispensable component in online services during recent years. Effective recommendation is essential for improving the services of various online business applications. However, serious privacy concerns have been raised on recommender systems requiring the collection of users' private information for recommendation. At the same time, the success of e-commerce has generated massive amounts of information, making scalability a key challenge in the design of recommender systems. As such, it is desirable for recommender systems to protect users' privacy while achieving high-quality recommendations with low-complexity computations.

This paper proposes a scalable privacy-preserving item-based top-N recommendation solution, which can achieve high-quality recommendations with reduced computation complexity while ensuring that users' private information is protected. Furthermore, the computation complexity of the proposed method increases slowly as the number of users increases, thus providing high scalability for privacy-preserving recommender systems. More specifically, the proposed approach consists of two key components: (1) MinHash-based similarity estimation and (2) client-side privacy-preserving prediction generation. Our theoretical and experimental analysis using real-world data demonstrates the efficiency and effectiveness of the proposed approach.
\end{abstract}

\section{Introduction}
\label{sec:intro}
With the explosive growth of information on the Internet, recommender systems have become indispensable in many online services~\cite{Adomavicius05}. As one of the most important techniques in recommender systems, collaborative filtering (CF)~\cite{he2017modeling,li2012interest} has been broadly adopted, such as in Amazon~\cite{Li2016311,Amazon}, MovieLens~\cite{Canny02}, YouTube~\cite{davidson2010youtube}, and so on. Such broad applications of CF raise serious concerns about the leakage of users' privacy.

The privacy concern originates from the basic idea behind CF techniques. For instance, item-based top-N recommendation, as one of the widely used CF techniques, assumes that a user may be interested in items that are similar to the items that he/she liked before.
More specifically, the problem of top-N recommendation aims to provide an ordered list of $N$ items to a user. Item-based recommendation works by first collecting users'  ``ratings'' of items. Then, items are profiled by analyzing users' ratings for them. At last, a user is recommended with  items that have a high similarity with his/her historically highly-rated items.  Through the recommendation, users can find items of their interest, such as movies, music, or other things. However, to profile items and perform recommendations for users, sensitive information, such as user demographics and user ratings, are collected by recommender systems, giving rise to serious privacy concerns of individual users~\cite{Li2016311,Adomavicius05,Li2017440,li2011yana,li2011pistis}. In addition, profiling and recommending items based on all users' information lead to scalability issues because of the explosive growth of online information.

Recent research works have aimed to tackle the privacy-preserving issues for individual users based on CF~\cite{ozturk2015existing}.
Scalability is the primary limiting factor to existing cryptography-based methods~\cite{Li2016311,Canny02,Esma08,Kikuchi09}. These methods adopt cryptography, such as homomorphic encryptions, to hide true user ratings during computation. These solutions require computation-intensive encryption and decryption operations. Hence they do not scale well given the large number of users and items~\cite{Canny02,Esma08,Kikuchi09}. Compared with cryptography-based solutions, random perturbation based methods do not conduct intensive computations but protect users' privacy by perturbing users' ratings, such as adding random noise, before sending the users' information to the server~\cite{Polat03,Zhang06,McSherry09}. As such, random perturbation based methods are efficient and easy to implement. However,  these methods trade accuracy for privacy, yet the protection of user privacy is not guaranteed. As shown in~\cite{Huang05}, the server can partially recover valid user data from perturbed data using learning techniques.

An ideal privacy-preserving recommender system should guarantee user privacy protection without compromising recommendation accuracy or efficiency. However,  existing privacy-preserving CF methods trade either efficiency (cryptography-based methods) or accuracy (random perturbation based methods) for privacy. Therefore, the focus of this work is to address the above challenges and develop a scalable solution capable of preserving privacy while achieving high-quality recommendations with low-complexity computations.

Our work is motivated by one key observation -- the overloading information generated on the Internet has a degree of redundancy for recommendations, as some users have a similar preference. As such,  it is possible to develop a scalable approach by collecting some anonymous users' information while achieving high-quality recommendations. Therefore, this work proposes an approach for item-based top-N recommendations, which can cope with a large number of user and item scenarios while achieving efficient privacy-preserving recommendation. Different from the methods above, the proposed approach neither performs additional calculation nor changes the original users' information to preserve users' privacy. Specifically, our proposed approach works as follows: (1) to protect individual users' privacy, we first use anonymous random walks to collect users' information, thus eliminating the correspondence between individual users and their information; (2) item similarities are estimated online using the MinHash-based method through the received information and the recommendations are generated locally; and (3) by restricting the number of random walks, the computation cost is reduced significantly, especially for scenarios with a large number of users and/or items.

The contributions of this work are summarized as follows:
\begin{enumerate}
\item This work proposes a scalable algorithm for privacy-preserving item-based top-N recommendations, which can achieve high-quality recommendations with low-complexity computations. Compared with existing methods, the computation time of the proposed method increases slowly as the number of users increases, thus providing high scalability for privacy-preserving recommender systems.

\item The evaluation results in three real-world data sets demonstrate that the proposed method can be more efficient than non-privacy-preserving item-based top-N recommendation methods. Specifically, when the accuracy loss is less 1.0\%, the corresponding computation time is reduced by 20.99\%, 57.35\%, and 62.81\% on the Last.fm data set, on the Jester data set, and on the MovieLens 20M data set, respectively.
\end{enumerate}

The rest of this paper is organized as follows.
Section~\ref{sctn::rltd} surveys the related works.
Section~\ref{sctn::prblm} describes the problem formulation.
Section~\ref{sctn::alg} presents the proposed scalable algorithm for privacy-preserving item-based top-N recommendation.
Section~\ref{sctn::exp} discusses experimental results on the real-world data sets.
Finally, Section~\ref{sctn::cnclsn} concludes this work.

\section{Related Work}
\label{sctn::rltd}
Existing privacy-preserving recommender systems can be classified into two main categories: cryptography based methods~\cite{Li2016311,Canny02,Esma08,Kikuchi09} and random perturbation based methods~\cite{Polat03,Zhang06,McSherry09}.

\subsection{Cryptography Based Methods}
Cryptography based methods adopt cryptography, such as homomorphic encryptions, to hide true user ratings during computation.

Kikuchi~\etal proposed using homomorphic encryption to calculate user similarities, item recommendation scores and decrypting the scores by a set of trusted authors~\cite{kikuchi2009privacy}. Since their method has all users' original data involved in the calculation, users can obtain recommendations without privacy violation. A similar method has also been introduced in~\cite{zhan2010privacy}, in which Zhan~\etal constructed a more efficient privacy-preserving collaborative recommender system based on the scalar product protocol by comparing with major cryptology approaches. Canny~\etal proposed an algorithm whereby a community of users can compute a public ``aggregate" of their data that does not expose individual users' data~\cite{Canny02}. They used homomorphic encryption to allow sums of encrypted vectors to be computed and decrypted without exposing individual users' data. 
These homomorphic encryption based approaches are computation inefficient because all computations are performed on encrypted data. Additionally, cryptography based methods suffer from the scalability issue since encryption and decryption operations are computation intensive and do not scale well given a large number of users and items.

\subsection{Random Perturbation Based Methods}
Random perturbation based methods perturb users' ratings to prevent the server from obtaining true users' ratings~\cite{Polat03,Zhang06,McSherry09}.

Polat~\etal proposed a privacy-preserving CF method to perturb individual user's original data by adding a random number, while accurately estimating the aggregation data from a large number of users~\cite{Polat03}. Casino~\etal proposed a $k$-anonymous approach to protect the user's privacy, in which they clustered similar users and profiled the clusters. Then the profile in the same cluster had the capability of representing the users in the cluster, achieving perturbing individual users' data and preserving users' privacy. Storing users' profiles in a distributed manner and achieving recommendations is another option for perturbation based privacy-preserving recommender systems. Shokri~\etal proposed a distributed mechanism to increase privacy-preserving while achieving accurate recommendations~\cite{shokri2009preserving}. In their methods, users first store their profiles on their own sides (called offline profiles). Then, the offline profiles are partly merged with the profiles of similar users. After that, the offline profiles are uploaded to a central server periodically for participating in recommendations. A different approach is presented in~\cite{boutet2016privacy} that relies on a two-fold mechanism implemented on a decentralized user-based CF. In the work, users' exact profiles are prevented from exchanging on their own sides while constructing interest-based topology, through which obfuscated profiles are still extracted, ensuring the quality of recommendations. Random perturbation based methods trade accuracy for privacy, yet the protection of user privacy is not guaranteed.
As shown in~\cite{Huang05}, the server can partially recover true user data
from perturbed data using learning techniques. 

\section{Problem Formulation}
\label{sctn::prblm}
\subsection{Privacy Issue in Item-based Top-N Recommendation}
The main purpose of recommender systems is to estimate user ratings
of items that have not been seen by the user~\cite{Adomavicius05}.
Generally, recommender systems estimate users' ratings for a specific
item based on users' previous ratings on other items. Thus, recommender
systems need to collect users' item ratings before estimating users'
ratings of unrated items. This is a violation of user privacy, since most users
do not wish to expose their item ratings~\cite{Berkovsky07}.
Sensitive user information (e.g., user preferences) may be collected,
analyzed, and even sold when a company declares bankruptcy~\cite{Canny02}. Ideally, users of recommender systems should be able to obtain
high-quality recommendations efficiently, without exposing their
item ratings to the recommendation server or any other third parties.
Generally, two key steps that are required
in item-based top-N recommendation~\cite{Deshpande04} are described in the following two subsections.
\subsection{Item Similarity Computation}
A variety of methods can be adopted to calculate the similarities between items, such as cosine similarity~\cite{Sarwar,tan2017efficient}, correlation-based similarity~\cite{Sarwar} and Jaccard similarity~\cite{Das07}, etc.
For binary ratings (i.e., 0 means dislike and 1 means like),
Jaccard similarity is the most widely adopted similarity measure.
Given two items $i$ and $j$, their Jaccard similarity is defined as follows~\cite{Das07}:

\begin{equation}
\label{eqn:sim}
	sim_{i,j}= Jaccard(U_{i},U_{j})= \dfrac{|U_{i}\bigcap U_{j}|}{|U_{i}\bigcup U_{j}|},
\end{equation}
where $U_{i}$ ($U_{j}$) is the set of users who like item $i$ ($j$).
Computing the Jaccard similarity is time-consuming, especially when the number of
items is large. Thus, cryptography based privacy-preserving
recommendation is not practical for large datasets.

\subsection{Predicting Item Ratings and Top-N Recommendation}
After obtaining item similarities, predictions for unrated items are computed
by taking a weighted average of a target user's past item ratings.
Given a target user $u$ and an unrated item $i$, the predicted rating of $u$ on $i$
can be computed as follows~\cite{Deshpande04}:
\begin{equation}
\label{eqn:score}
\hat{r}_{u,i}= \dfrac{\sum_{j\in I_u}r_{u,j}*sim_{i,j}}{\sum_{j\in I_u}sim_{i,j}},
\end{equation}
where $I_u$ is the set of items rated by $u$ and $r_{u,j}$ is $u$'s rating
on item $j$. 

However, the rating $r_{u,j}$ is 1 if the user $u$ like the $j$th item and 0 otherwise for binary ratings. The predicted rating $r_{u,j}$ calculated by Equation~\ref{eqn:score} is always a constant for the $j$th item. 
Therefore, for binary ratings, the predicted rating of $u$ on $i$
can be rewritten as Equation~\ref{eqn:score2}~\cite{Das07}:
\begin{equation}
\label{eqn:score2}
\hat{r}_{u,i}= \sum_{j\in I_u}r_{u,j}*sim_{i,j}.
\end{equation}

After computing $\hat{r}_{u,i}$ for all unrated items of $u$, the similarity ranking of the unrated items of user $u$ can be obtained, and the items with the highest $N$ predicted ratings will be recommended to $u$. 

\section{Scalable Privacy-Preserving Item-based Top-N Recommendation}
\label{sctn::alg}
This section first gives the overview of the proposed scalable privacy-preserving item-based Top-N recommendation approach, then details the process of the approach. After that, the theoretical analysis of the efficiency and effectiveness of the proposed approach is presented.

\subsection{Solution Overview}
In this work, our goal is to protect user privacy and reduce computation complexity with minimum loss in recommendation accuracy.
We achieve this by proposing the following techniques:
1) a MinHash-based privacy-preserving similarity estimation method, which can estimate the Jaccard similarity of items with high efficiency and protect user privacy during the computation process; and 2) a client-side privacy-preserving prediction generation method, which can predict item ratings for users with privacy protection. Unlike existing works, which trade
either efficiency (e.g., cryptography-based methods) or accuracy (e.g., perturbation-based methods) for privacy, our solution can guarantee privacy protection while supporting flexible balancing between accuracy and efficiency. The flowchart of the proposed approach is shown as Figure~\ref{fig:overview}.

\begin{figure}[h]
	\centering
	\includegraphics[width=0.50\textwidth]{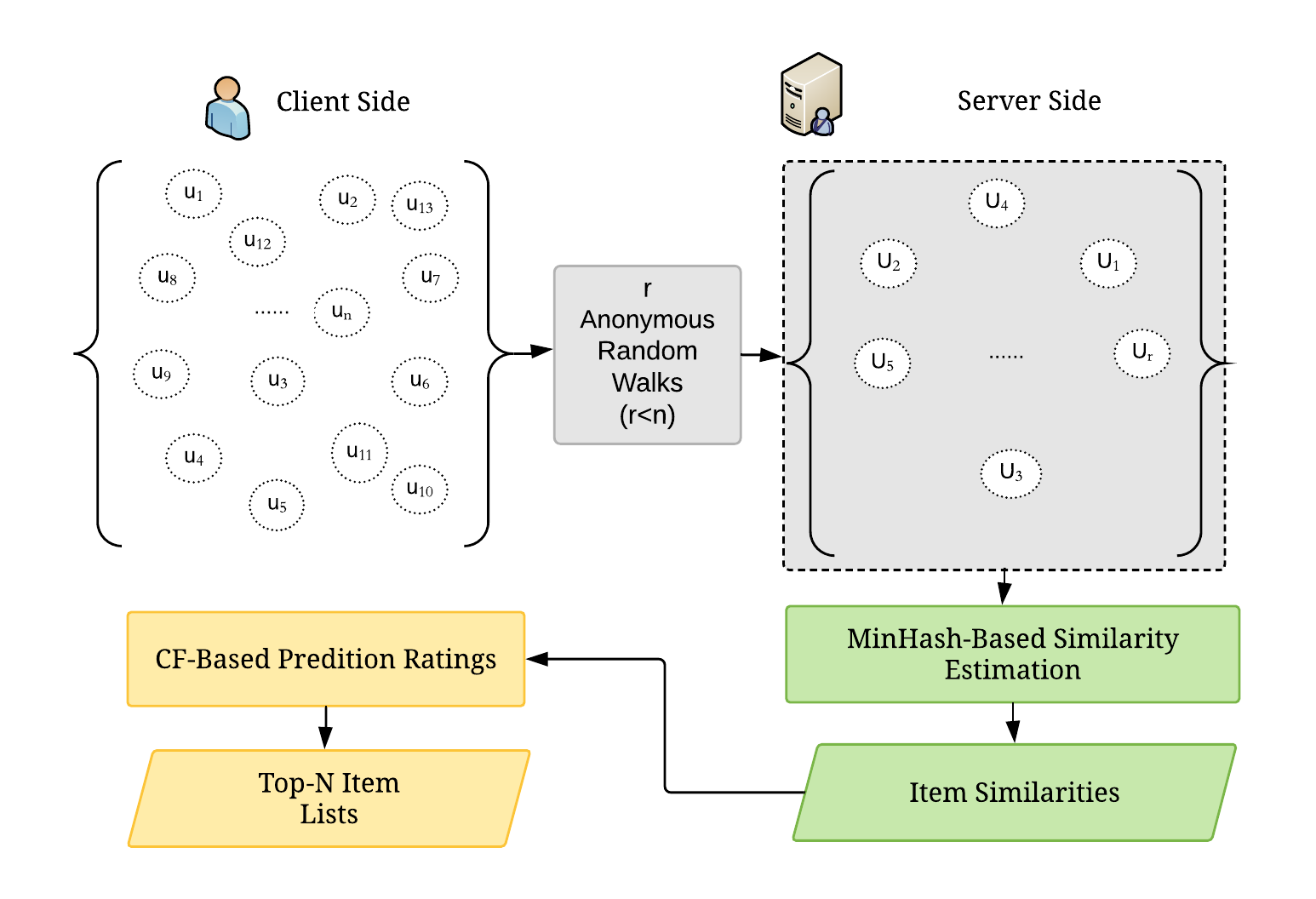}
	\caption{Flowchart of the proposed privacy-preserving top-$N$ recommendation. }
	\label{fig:overview}
\end{figure}

\subsection{MinHash-based Privacy-Preserving Jaccard Similarity Estimation}

MinHash is an efficient method for estimating Jaccard similarity between
two sets~\cite{minhash}. Let $\mathcal{H}=\{h_1,h_2,\cdots,h_k\}$ denote
$k$ independent random perturbations on elements (hash functions), for any two sets $S_i$ and $S_j$,
they have the following property~\cite{minhash}:
\begin{equation}
\label{eqn:jaccard}
{\bf Pr}_{h\in \mathcal{H}} \big[h(S_i)==h(S_j)\big] = \frac{|S_i\cap S_j|}{|S_i\cup S_j|} = Jaccard (S_i,S_j).
\end{equation}
In Equation~\ref{eqn:jaccard}, the number of hash functions --- $k$ is the key
factor for determining estimation accuracy and efficiency, i.e.,
high $k$ values indicate higher estimation accuracy but lower
computation efficiency. Theoretically, if $k$ is large enough, then
the estimation can be as accurate as
computation using Equation~\ref{eqn:sim}. Tradeoffs between accuracy
and efficiency of the method are further discussed in Section~\ref{sec:analysis_accuracy}.

Based on MinHash, Jaccard similarities among item pairs can be efficiently
estimated, but user privacy should be strictly protected during this
process. To this end, a privacy-preserving MinHash protocol is proposed,
in which users choose to add their information in anonymous random walks.
Thus, the server and any other user cannot know which piece of data
are from a target user and whether a target user has added his/her data,
so that user privacy can be protected. The detailed procedure for
estimating Jaccard similarity based on privacy-preserving MinHash method
is presented in Algorithm~\ref{alg:hash}.

\begin{algorithm}[h!]%
\caption{PrivateJaccard($U$, $I$, $\mathcal{H}$)}
\label{alg:hash}
\begin{algorithmic}[1]
\Require $U$ and $I$ are the sets of users and items, and $\mathcal{H}$ is the set of hash functions.
\ForAll {$i,j\in I (i\neq j)$}
    \State $n_{i,j} = 0$; ($n_{i,j}$ records the number of times that hash values of $i$ and $j$ are equal)
\EndFor
\For {each $h\in \mathcal{H}$}
    \State Let $I_{h} = \emptyset$, the server randomly selects $u\in U$ as $u^*$;
    \While{$I_{h}$ is not sent to server}
        \State $I_{u^*}$ is the items liked by $u^*$, and $\rho_{u^*}\in (0,1)$ is its predefined probability;
        \State $u^*$ randomly generates $0\le\rho\le 1$;
        \If {$\rho<\rho_{u^*}$}
            \If {$I_{h} == \emptyset$}
                \State $I_{h} = I_{u^*}$;
            \Else
                \State $u^*$ sends $I_{h}$ to the server; Break;
            \EndIf
        \Else
            \State $u^*$ randomly chooses $u' \in U$, and sends $I_{h}$ to $u'$;
            \State $u^*\gets u'$;
        \EndIf
    \EndWhile
    \ForAll {$i,j\in I_{h} (i\neq j)$}
        \State $n_{i,j} \gets n_{i,j} + 1$;
    \EndFor
\EndFor
\ForAll {$i,j\in I (i\neq j)$}
    \State $Jaccard(i,j) = n_{i,j}/|\mathcal{H}|$;
\EndFor
\end{algorithmic}
\end{algorithm}

\subsection{Client-Side Privacy-Preserving Prediction Generation}
Based on Algorithm~\ref{alg:hash}, the recommender system
can obtain the Jaccard similarities among all item pairs. Then,
the server can send the item similarities to all users. After obtaining
item similarities, each client can compute its own prediction scores based
on Equation~\ref{eqn:score2} and recommend items with high $\hat{r}$ values
to its user. Since the computation of this step is fully accomplished on the
client side, user privacy can be strictly protected.

\subsection{Discussion}
\subsubsection{Complexity Analysis}
\label{sec:analysis_complexity}
The complexity of Jaccard similarity computation is $O(m^2n)$,
where $m$ is the number of items and $n$ is the number of users. But
based on Algorithm~\ref{alg:hash}, the computation complexity is reduced
to $O(k+m^2)$, where $k$ is the number of hash functions. This is because
the server can go through the $k$ hash results, and find out all the cases
that $h(i)==h(j)$. Then, the server can calculate the probability of $h(i)==h(j)$
for all $m^2/2$ item pairs. For prediction generation, the server side computation
is zero. And the computation complexity for each user is $O(lm)$, where $l$
is the number of items rated by the user.

\subsubsection{Accuracy of Similarity Estimation}
\label{sec:analysis_accuracy}
The similarity estimation accuracy of Algorithm~\ref{alg:hash}
is determined by $k$, the number of hash functions. Here, we analyze
the relationship between estimation accuracy and $k$ theoretically;
detailed statistical results are presented in the evaluation section.
We first introduce the Chernoff Bound before analyzing the accuracy of the method.

\begin{theorem}[Chernoff Bound~\cite{Chernoff}]
Given a set of r independent identically distributed (iid) random variables
${X_1,X_2,\ldots,X_r}$, satisfying that $-\Delta\leq X_i \leq \Delta$ and $E[X_i]=0$
$(i\in[1,\dots,r])$. Let $M={\mathop{\sum}^{r}_{i=1}X_i}$. Then for any $\alpha\in(0,\frac{1}{2})$,
$Pr[|M|>\alpha]\leq2\exp(\frac{-\alpha^2}{2r\Delta^2})$.
\end{theorem}


\begin{theorem}
\label{thm:similarity}
Given a set of hash functions ${h_1, h_2,\cdots, h_k}$ and items $I$,
$\forall i,j \subseteq I$ $(i\neq j)$, let $\hat{J}(U_i,U_j)$ and $J(U_i,U_j)$
be the estimated and true Jaccard similarity between $i$ and $j$,
then for any $\alpha\in(0,\frac{1}{2})$ and $k=(2/\alpha ^2)\ln (2/\delta)$,
$Pr[|\hat{J}(U_i,U_j)-J(U_i,U_j)| \leq \alpha]\geq 1- \delta$.
\end{theorem}

\begin{proof}
For any $l\in {[1,2,\cdots, k]}$, we have $h_l= 1$ if $h_l(U_i)=h_l(U_j)$.
Let $x_l = (1/k)(h_l-J(U_i,U_j))$. Then, $x_1, x_2, \cdots, x_k$ is a set of $k$
independent identically distributed random variables and $-\frac{1}{k} \leq x_l \leq \frac{1}{k}$.
Let $M=\sum^k_{l=1}x_l$, then $E[M] = 0$ (based on Equation~\ref{eqn:jaccard}).
Then, according to Chernoff Bound, let $\Delta = \frac{1}{k}$ and $r=k=(2/\alpha ^2)\ln (2/\delta)$,
we have $Pr[|\hat{J}(U_i,U_j)-J(U_i,U_j)| > \alpha]\leq \delta$, i.e.,
$Pr[|\hat{J}(U_i,U_j)-J(U_i,U_j)| \leq \alpha]\geq 1- \delta$.
\end{proof}

\subsubsection{Privacy Analysis}
\label{sec:analysis_privacy}
In the proposed method,
only item similarity computation, which are performed among users, may reveal
user privacy. Here, we prove that Algorithm~\ref{alg:hash} can strictly 
protect user privacy under the  semi-honest model~\cite{Goldreich},
in which users follow the computation protocols honestly except that they
can infer information based on intermediate data.
The ``privacy'' definition is adopted from Goldreich~\cite{Goldreich},
which  states
that a computation protocol is privacy-preserving if the view of each party
during the execution of the protocol can be simulated by a polynomial-time
algorithm knowing only the input and the output of the party.

\begin{theorem}
Algorithm~\ref{alg:hash} is privacy-preserving for users under the semi-honest model.
\end{theorem}
\begin{proof} The simulator for Algorithm~\ref{alg:hash} can be constructed as follows:
\begin{itemize}
\item {\emph{Stage 1:}} In this stage, the server randomly chooses $u^*$
to start a random walk or $u^*$ receives an empty $I_{h_k}$.
If $u^*$ decides not to add its data, the output of $u^*$ is an empty
$I_{h_k}$ which can be easily simulated. Otherwise, the simulator for $u^*$
can simulate $I_{h_k}$ with $I_{u^*}$, and then send $I_{h_k}$ to another user $u'$.
From the view of $u'$, $Pr(I_{h_k}==I_{u^*})=1/n$ ($n$ is the number of users),
i.e., $I_{h_k}$ can be from any user in $U$. Thus, $u'$ cannot learn any
information from $I_{h_k}$, no matter $u^*$ added his data or not.
This indicates that the output of $u^*$ is indistinguishable from
what the next user views in real random walk.
\item {\emph{Stage 2:}} In this stage, a user $u^*$ receives a non-empty
$I_{h_k}$. The simulator for $u^*$ can simulate his output with $I_{h_k}$.
No matter to whom $u^*$ chooses to send $I_{h_k}$ (the server or another user),
$I_{h_k}$ does not contain any private information of $u^*$ and the output
of $u^*$ is also indistinguishable from what the next party views in real random walk.
\end{itemize}

The above simulator is linear in the size of $I_{h_k}$, i.e.,
a polynomial-time simulator is successfully constructed for users.
Thus, Algorithm~\ref{alg:hash} is privacy-preserving for users.
\end{proof}

\section{EXPERIMENTS}
\label{sctn::exp}
This section evaluates the efficiency and effectiveness of the proposed approach on three real-world data sets. The first study evaluates how the recommendation accuracy of the proposed method is affected by the key factor $k$ in the proposed algorithm. Then, the absolute error of similarity estimation is further analyzed to help understand the proposed method. 
\subsection{Experimental Setup}
The evaluation data sets are collected from three real-world data sets that have been widely used for evaluating recommendation algorithms. Table~\ref{tb:ds} shows the three data sets in detail.
\begin{table}[h]
    \scriptsize
    \centering
    \caption{Description of Data Sets}
    \label{tb:ds}
    \begin{tabular}{c|c|c|c}
        \hline
        \textbf{Data Set}& \textbf{\# of Users} & \textbf{\# of Items}  & \textbf{Density (\%)}                   \\\hline
         Last.fm    & 17,976 & 8,007  &0.65 \\\hline
        Jester   & 24,983 & 100  &41.97  \\\hline
        MovieLens 20M   & 7,120&131,262 &0.11 \\ \hline
    \end{tabular}
\end{table}
It is necessary to mention that this study changes the ratings from a real number to a binary number in the Jester data set and the MovieLens 20M data set. For instance, the ratings are from -10 to 10 in Jester data set and from 0 to 5 in the MovieLens 20M data set. We change the ratings to 1 if a user has rated an item, and to 0 otherwise. 

For each data set, we split it into train and test sets randomly by setting the ratio between the train set and the test set as 4:1. The results are presented by averaging the results of ten different random train-test splits.
Since the privacy property of the proposed method has been proved theoretically, the evaluation focuses on comparing the efficiency and
the accuracy of the proposed method (PP-IBTN) and an item-based top-N recommendation algorithm~\cite{Deshpande04,li2014item}
(IBTN). 
\subsection{Evaluation Metrics}
This study adopts precision metrics to evaluate the accuracy of the proposed approach, which is defined as follows:
\begin{equation}
\label{eqn:precison}
Precision= \dfrac{|U_i\cap U_r|}{|U_r|},
\end{equation}
where $U_i$ is the set of items that a user rated and $U_r$ is the set of items that are recommended.

Also, Equation~\ref{eqn:ae} defines the absolute error of similarity estimation:
\begin{equation}
\label{eqn:ae}
AE= |\hat { sim_{i,j}} - sim_{i,j}|,
\end{equation}
where $\hat { sim_{i,j}}$ is the similarity between item $i$ and $j$ calculated by the proposed Algorithm~\ref{alg:hash}.

\begin{figure*}
\centering

\begin{minipage}{5.5cm}
\subfloat[Recommendation Efficiency vs. $k$.]{\includegraphics[width=6cm]{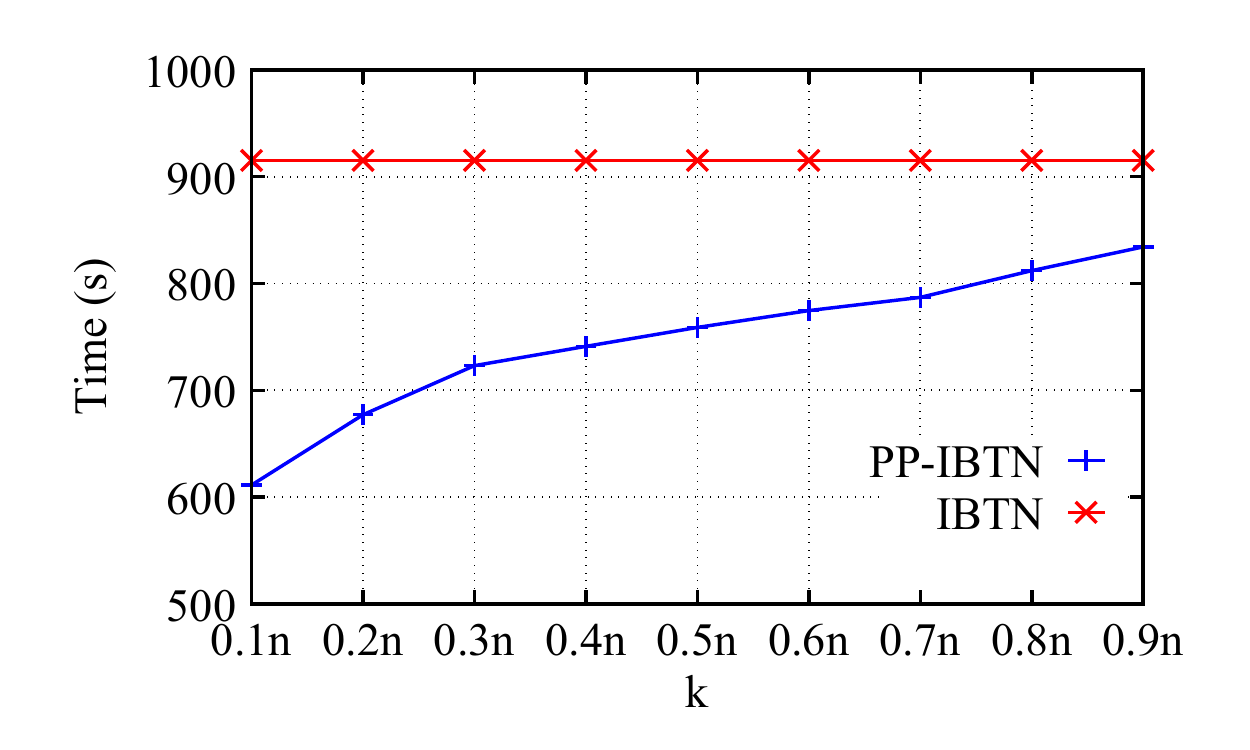}
  \label{fig:efficiency_lf}}
\end{minipage}
\begin{minipage}{5.5cm}
\subfloat[Precision Loss vs. $k$.]{\includegraphics[width=6cm]{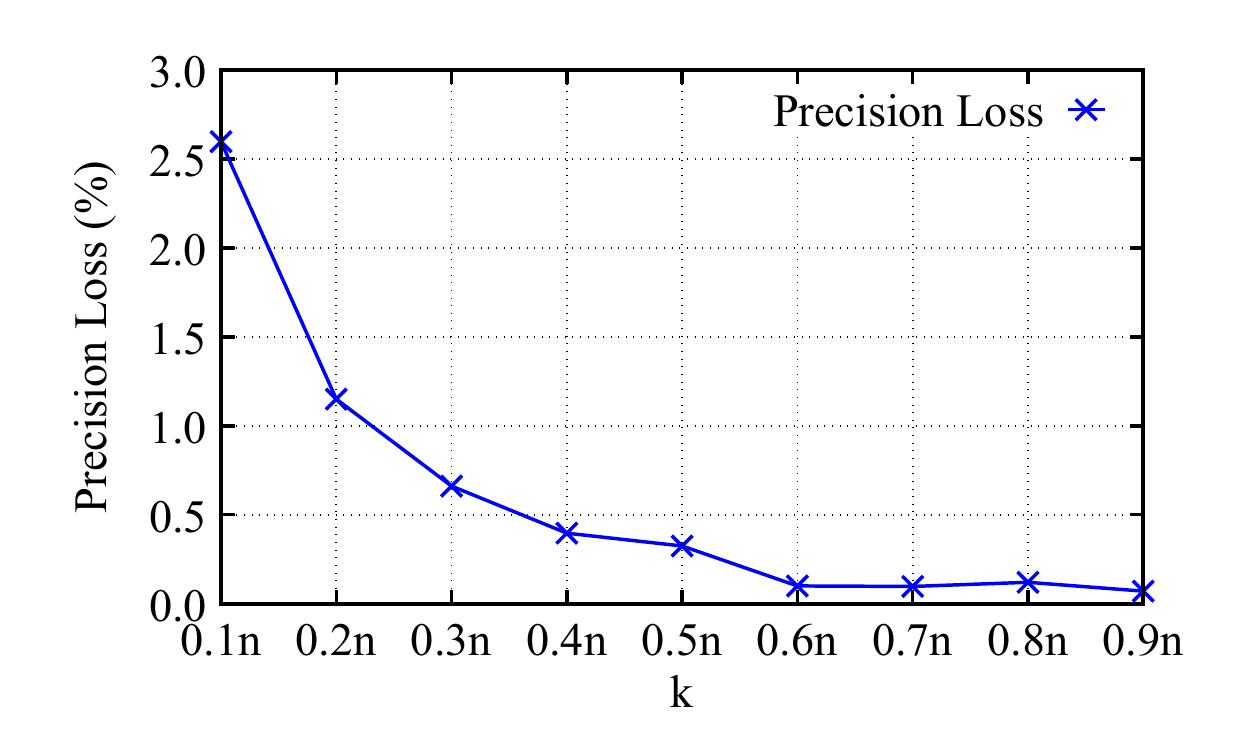}
 \label{fig:precision_lf}}
\end{minipage}
\begin{minipage}{5.5cm}
\subfloat[Similarity Estimation Accuracy.]{\includegraphics[width=6cm]{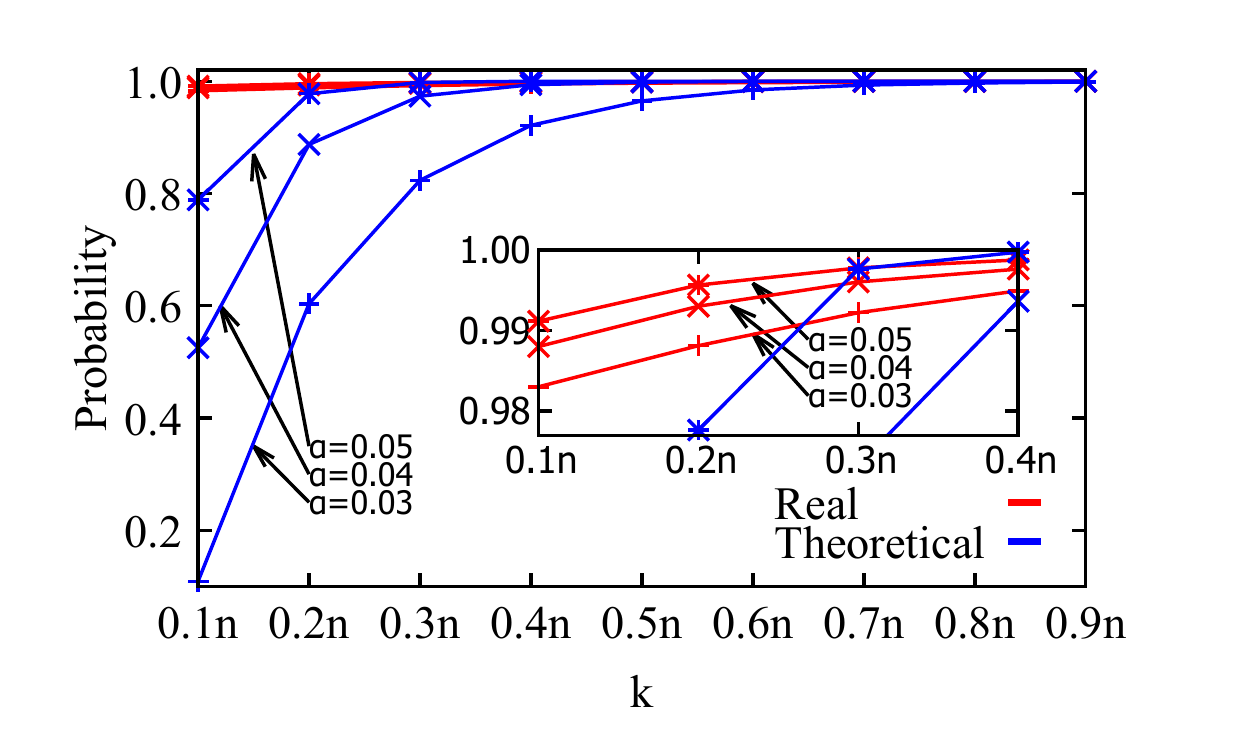}
 \label{fig:similarity_lf}}
\end{minipage}
\caption{Recommendation on Last.fm.}
\end{figure*}
\begin{figure*}
\centering
\begin{minipage}{5.5cm}
\subfloat[Recommendation Efficiency vs. $k$.]{\includegraphics[width=6cm]{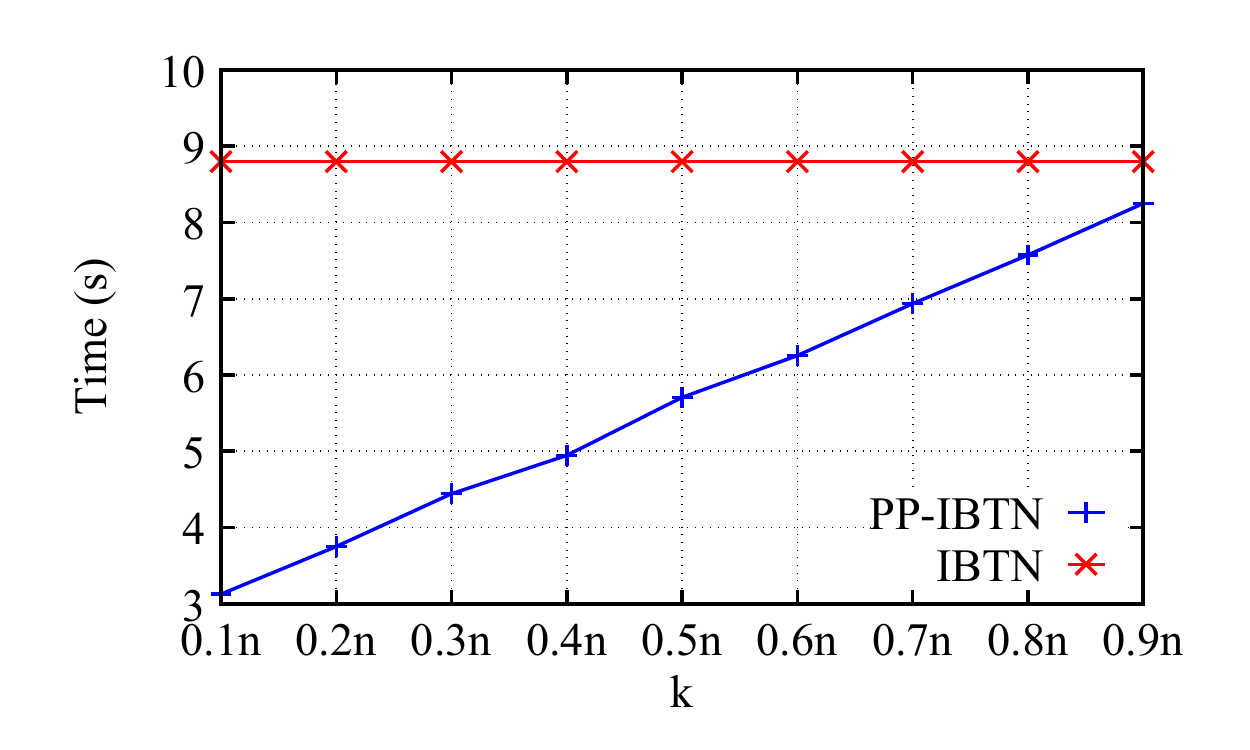}
  \label{fig:efficiency_je}}
\end{minipage}
\begin{minipage}{5.5cm}
\subfloat[Precision Loss vs. $k$.]{\includegraphics[width=6cm]{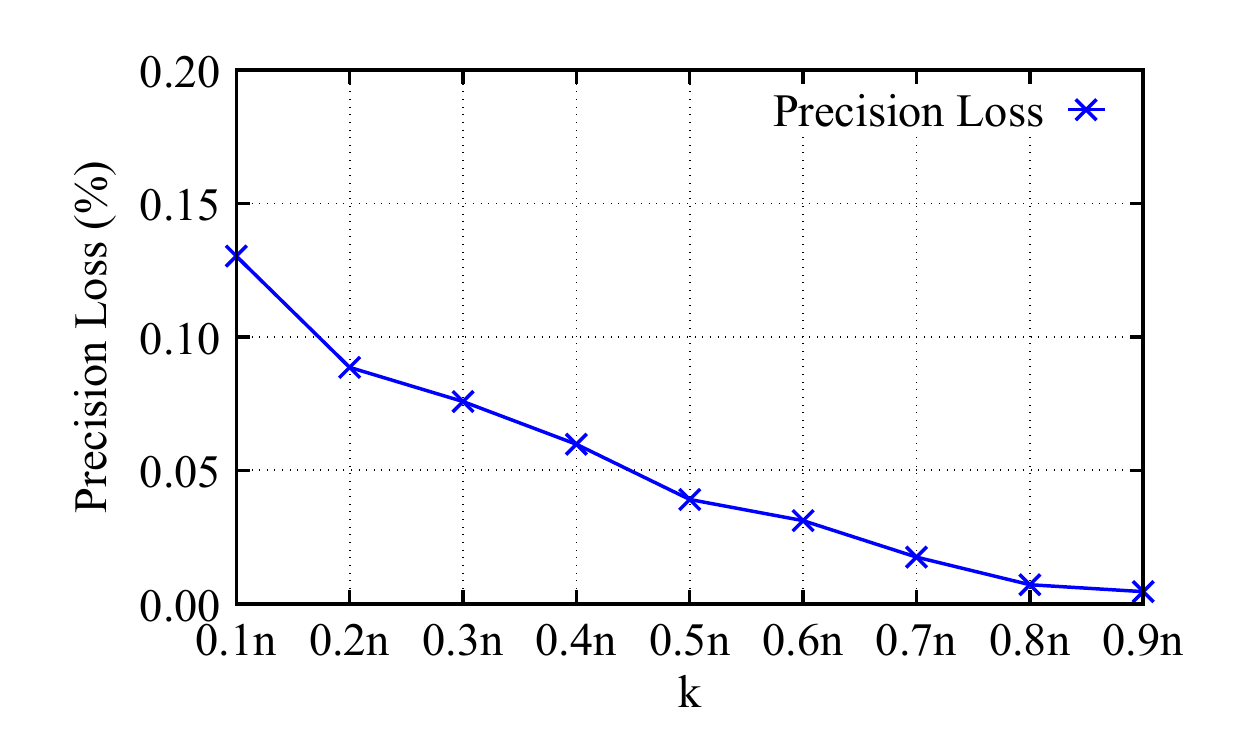}
 \label{fig:precision_je}}
\end{minipage}
\begin{minipage}{5.5cm}
\subfloat[Similarity Estimation Accuracy.]{\includegraphics[width=6cm]{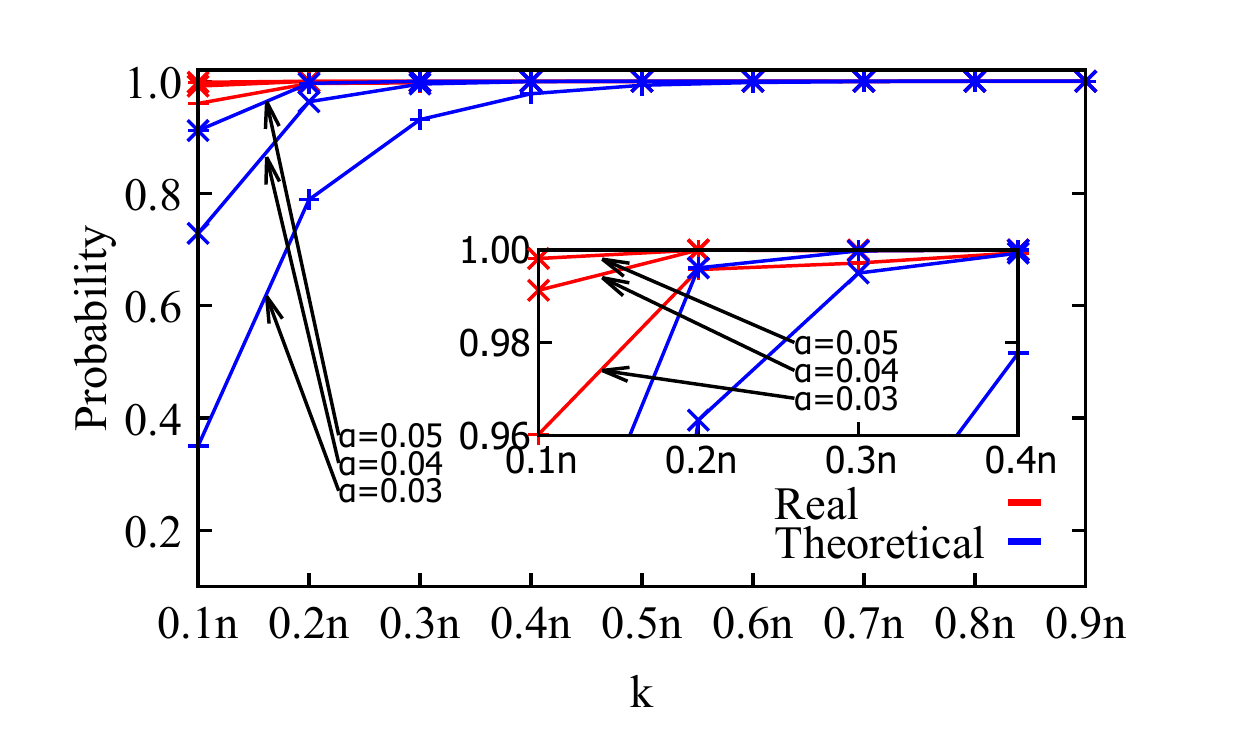}
 \label{fig:similarity_je}}
\end{minipage}
\caption{Recommendation on Jester.}
\end{figure*}

\begin{figure*}
\centering
\begin{minipage}{5.5cm}
\subfloat[Recommendation Efficiency vs. $k$.]{\includegraphics[width=6cm]{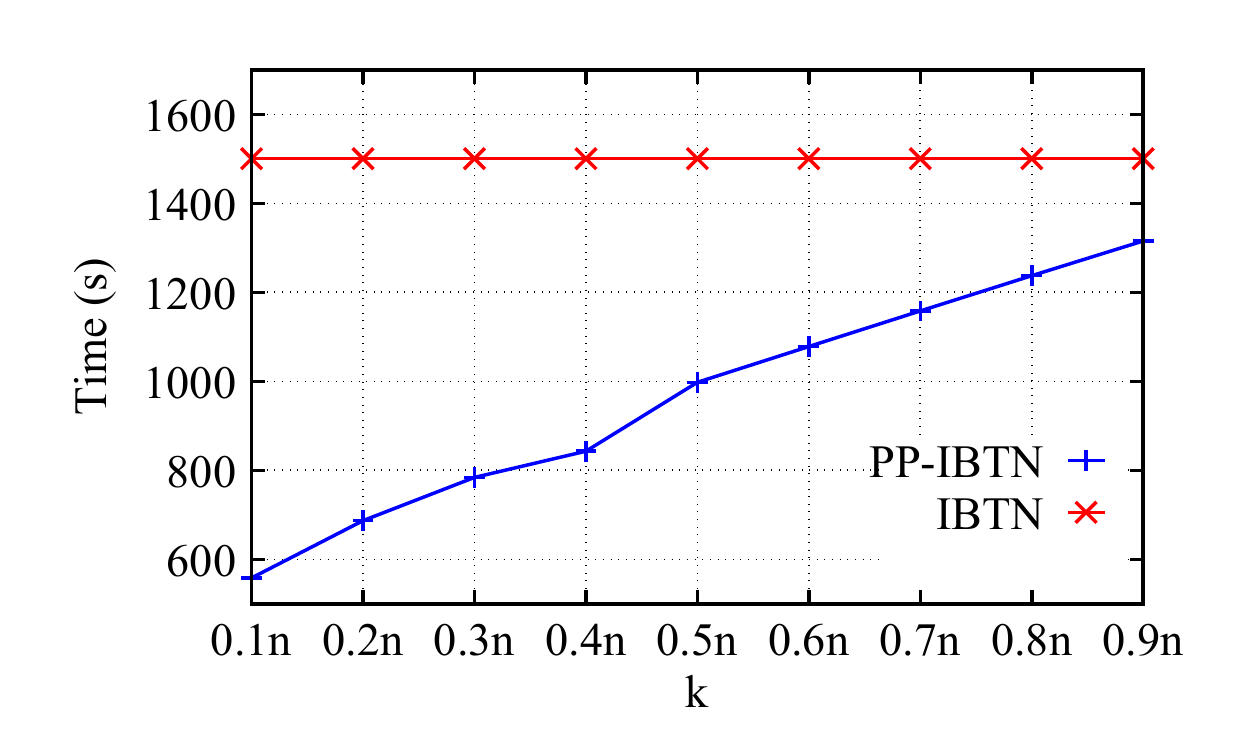}
  \label{fig:efficiency_ml}}
\end{minipage}
\begin{minipage}{5.5cm}
\subfloat[Precision Loss vs. $k$.]{\includegraphics[width=6cm]{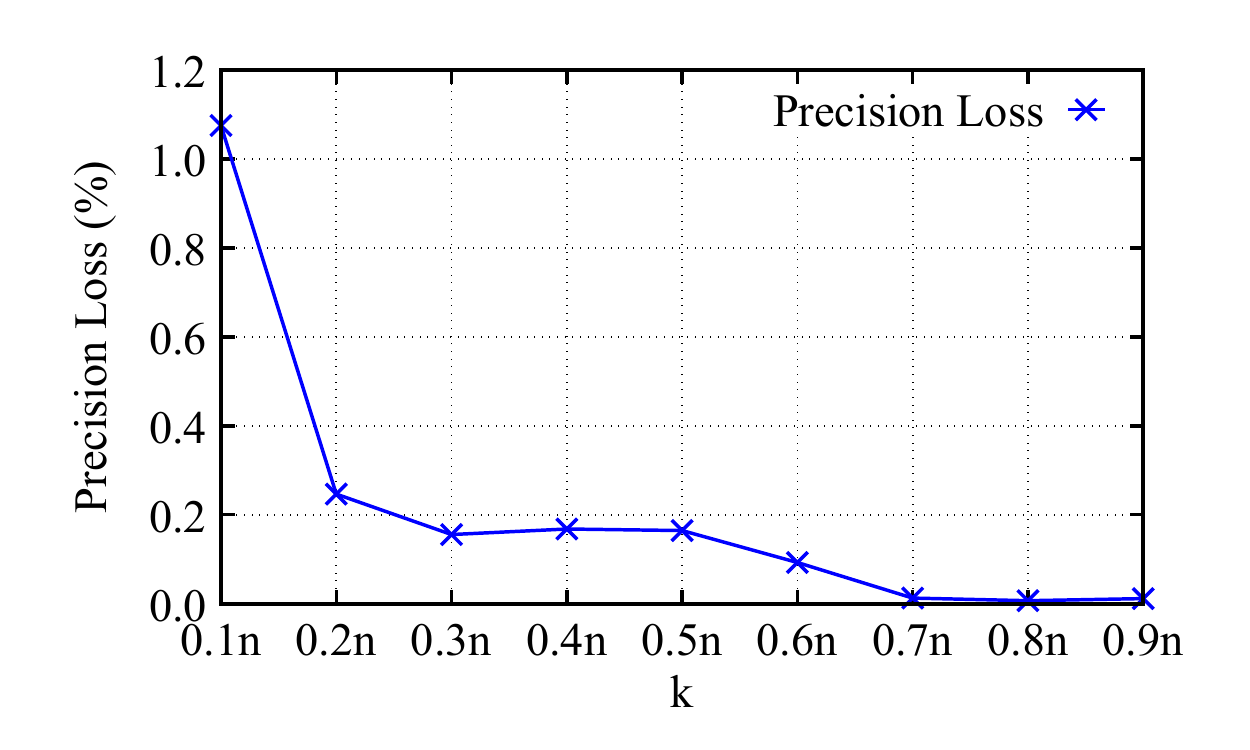}
 \label{fig:precision_ml}}
\end{minipage}
\begin{minipage}{5.5cm}
\subfloat[Similarity Estimation Accuracy.]{\includegraphics[width=6cm]{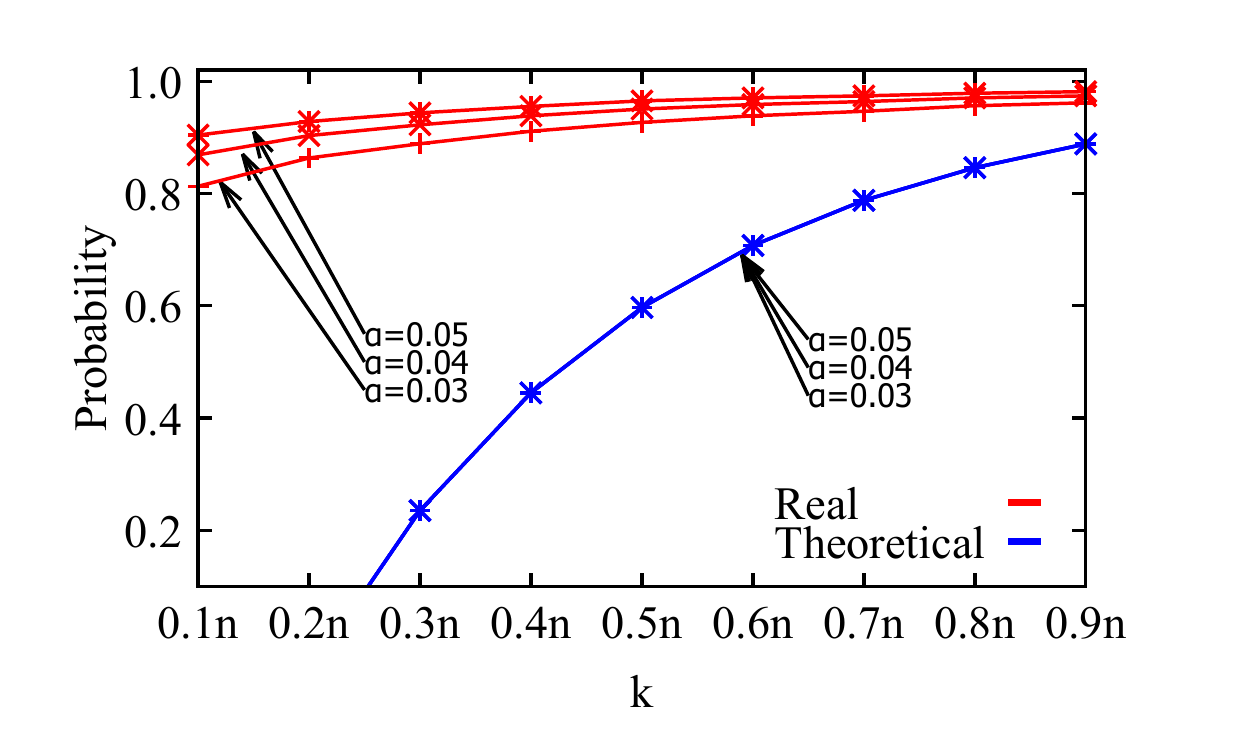}
 \label{fig:similarity_ml}}
\end{minipage}
\caption{Recommendation on MovieLens 20M.}
\end{figure*}

\subsection{Recommendation Efficiency Comparison}
Figure~\ref{fig:efficiency_lf}, Figure~\ref{fig:efficiency_je}, and Figure~\ref{fig:efficiency_ml} show the recommendation efficiency
of the proposed PP-IBTN method and the IBTN method on Last.fm, Jester, and MovieLens 20M, respectively, where $n$ is the number of users
and $k$ is the number of hash functions. 
We can see that for all
$k$ values, the computation
time of the PP-IBTN method is less than that of the IBTN method.
For instance, compared with the IBTN method, when $k=0.3n$, the proposed PP-IBTN method
can reduce the computation time by approximately 20.99\%, 49.50\%, and 47.74\% on Last.fm, Jester, and MovieLens 20M, respectively.
And when $k=0.9n$, the computation time is reduced by approximately 8.85\%, 6.25\%, and 12.32\% on Last.fm, Jester, and MovieLens 20M, respectively. In the proposed PP-IBTN method,
the hash procedure can greatly reduce the density of the data set, thus
the similarity computation step is much more efficient than that of the IBTN method.

\subsection{Recommendation Accuracy Comparison}
Figure~\ref{fig:precision_lf}, Figure~\ref{fig:precision_je}, and Figure~\ref{fig:precision_ml} show the recommendation precision loss
of the proposed PP-IBTN method
relative to the IBTN method on Last.fm, Jester, and MovieLen 20M, respectively.
More specifically, when $k=0.1n$, precision loss is approximately 2.60\%, 0.13\%, and 1.07\% on Last.fm, Jester, and MovieLens 20M, respectively. Recommendation precision
loss decreases to 0 when $k$ increases to $n$ in the three data sets.
Especially, the recommendation precision losses are less than 1\% when $k$ is
larger than $0.3n$, which indicates that the proposed method can achieve
decent accuracy. 

\subsection{Accuracy Analysis of Similarity Estimation}
Figure \ref{fig:similarity_lf}, Figure \ref{fig:similarity_je}, and Figure \ref{fig:similarity_ml} show the probability that the absolute error of
similarity estimation is less than $\alpha$ on Last.fm, Jester, and MovieLens 20M, respectively. We can see from the results that
experimental probabilities are higher than theoretical bounds for $\alpha$ values of 0.03, 0.04, and 0.05, and the probabilities are closer to 1 when $k$ increases.
This confirms that the similarity estimation accuracy is bounded, as in Theorem~\ref{thm:similarity}.
Meanwhile, the results also indicate that recommender systems can choose
different $k$ values to balance between similarity estimation accuracy and efficiency.

\section{Conclusion}
\label{sctn::cnclsn}
Recommender systems have played an essential role in e-commerce in recent years. However, existing solutions for recommendation have limited capabilities when it comes to protecting user privacy while still achieving high scalability. In this work, we have proposed a scalable algorithm for privacy-preserving, item-based top-N recommendations. The proposed algorithm can guarantee the protection of user privacy while significantly enhancing recommendation efficiency with decent recommendation quality. Comprehensive theoretical and experimental analysis demonstrates the efficiency and effectiveness of the proposed approach.

\section*{Acknowledgment}
This work was supported in part by the National Natural Science Foundation of China under Grant No. 61233016, and the National Science Foundation (NSF) of United States under grant No. 1334351 and 1442971.

\bibliographystyle{IEEEtran}
\bibliography{reference}

\begin{thebibliography}{10}
\providecommand{\url}[1]{#1}
\csname url@samestyle\endcsname
\providecommand{\newblock}{\relax}
\providecommand{\bibinfo}[2]{#2}
\providecommand{\BIBentrySTDinterwordspacing}{\spaceskip=0pt\relax}
\providecommand{\BIBentryALTinterwordstretchfactor}{4}
\providecommand{\BIBentryALTinterwordspacing}{\spaceskip=\fontdimen2\font plus
\BIBentryALTinterwordstretchfactor\fontdimen3\font minus
  \fontdimen4\font\relax}
\providecommand{\BIBforeignlanguage}[2]{{%
\expandafter\ifx\csname l@#1\endcsname\relax
\typeout{** WARNING: IEEEtran.bst: No hyphenation pattern has been}%
\typeout{** loaded for the language `#1'. Using the pattern for}%
\typeout{** the default language instead.}%
\else
\language=\csname l@#1\endcsname
\fi
#2}}
\providecommand{\BIBdecl}{\relax}
\BIBdecl

\bibitem{Adomavicius05}
G.~Adomavicius and A.~Tuzhilin, ``Toward the next generation of recommender
  systems: a survey of the state-of-the-art and possible extensions,''
  \emph{IEEE Transactions on Knowledge and Data Engineering}, vol.~17, no.~6,
  pp. 734--749, June 2005.

\bibitem{he2017modeling}
Y.~He, C.~Wang, and C.~Jiang, ``Modeling data correlations in recommendation,''
  \emph{IEEE Access}, vol.~5, pp. 11\,030--11\,042, 2017.

\bibitem{li2012interest}
D.~Li, Q.~Lv, X.~Xie, L.~Shang, H.~Xia, T.~Lu, and N.~Gu, ``Interest-based
  real-time content recommendation in online social communities,''
  \emph{Knowledge-Based Systems}, vol.~28, pp. 1--12, 2012.

\bibitem{Li2016311}
D.~Li \emph{et~al.}, ``An algorithm for efficient privacy-preserving item-based
  collaborative filtering,'' \emph{Future Generation Computer Systems},
  vol.~55, pp. 311 -- 320, 2016.

\bibitem{Amazon}
G.~Linden, B.~Smith, and J.~York, ``Amazon. com recommendations: Item-to-item
  collaborative filtering,'' \emph{IEEE Internet computing}, vol.~7, no.~1, pp.
  76--80, 2003.

\bibitem{Canny02}
J.~Canny, ``Collaborative filtering with privacy,'' in \emph{Proceedings of
  IEEE Symposium on Security and Privacy}.\hskip 1em plus 0.5em minus
  0.4em\relax IEEE, 2002, pp. 45--57.

\bibitem{davidson2010youtube}
J.~Davidson \emph{et~al.}, ``The youtube video recommendation system,'' in
  \emph{Proceedings of the fourth ACM conference on Recommender systems}.\hskip
  1em plus 0.5em minus 0.4em\relax ACM, 2010, pp. 293--296.

\bibitem{Li2017440}
D.~Li \emph{et~al.}, ``Efficient privacy-preserving content recommendation for
  online social communities,'' \emph{Neurocomputing}, vol. 219, pp. 440 -- 454,
  2017.

\bibitem{li2011yana}
D.~Li, Q.~Lv, L.~Shang, and N.~Gu, ``Yana: an efficient privacy-preserving
  recommender system for online social communities,'' in \emph{Proceedings of
  the 20th ACM international conference on Information and knowledge
  management}.\hskip 1em plus 0.5em minus 0.4em\relax ACM, 2011, pp.
  2269--2272.

\bibitem{li2011pistis}
D.~Li, Q.~Lv, H.~Xia, L.~Shang, T.~Lu, and N.~Gu, ``Pistis: a
  privacy-preserving content recommender system for online social
  communities,'' in \emph{Web Intelligence and Intelligent Agent Technology
  (WI-IAT), 2011 IEEE/WIC/ACM International Conference on}, vol.~1.\hskip 1em
  plus 0.5em minus 0.4em\relax IEEE, 2011, pp. 79--86.

\bibitem{ozturk2015existing}
A.~Ozturk and H.~Polat, ``From existing trends to future trends in
  privacy-preserving collaborative filtering,'' \emph{Wiley Interdisciplinary
  Reviews: Data Mining and Knowledge Discovery}, vol.~5, no.~6, pp. 276--291,
  2015.

\bibitem{Esma08}
E.~A{\"\i}meur \emph{et~al.}, ``Alambic: a privacy-preserving recommender
  system for electronic commerce,'' \emph{International Journal of Information
  Security}, vol.~7, no.~5, pp. 307--334, 2008.

\bibitem{Kikuchi09}
H.~Kikuchi, H.~Kizawa, and M.~Tada, ``Privacy-preserving collaborative
  filtering schemes,'' in \emph{International Conference on Availability,
  Reliability and Security}.\hskip 1em plus 0.5em minus 0.4em\relax IEEE, 2009,
  pp. 911--916.

\bibitem{Polat03}
H.~Polat and W.~Du, ``Privacy-preserving collaborative filtering using
  randomized perturbation techniques,'' in \emph{The Third IEEE International
  Conference on Data Mining}.\hskip 1em plus 0.5em minus 0.4em\relax IEEE,
  2003, pp. 625--628.

\bibitem{Zhang06}
S.~Zhang, J.~Ford, and F.~Makedon, ``A privacy-preserving collaborative
  filtering scheme with two-way communication,'' in \emph{Proceedings of the
  7th ACM conference on Electronic commerce}.\hskip 1em plus 0.5em minus
  0.4em\relax ACM, 2006, pp. 316--323.

\bibitem{McSherry09}
F.~McSherry and I.~Mironov, ``Differentially private recommender systems:
  building privacy into the net,'' in \emph{Proceedings of the 15th ACM
  international conference on Knowledge discovery and data mining}.\hskip 1em
  plus 0.5em minus 0.4em\relax ACM, 2009, pp. 627--636.

\bibitem{Huang05}
Z.~Huang, W.~Du, and B.~Chen, ``Deriving private information from randomized
  data,'' in \emph{Proceedings of the 2005 ACM international conference on
  Management of data}.\hskip 1em plus 0.5em minus 0.4em\relax ACM, 2005, pp.
  37--48.

\bibitem{kikuchi2009privacy}
H.~Kikuchi, H.~Kizawa, and M.~Tada, ``Privacy-preserving collaborative
  filtering schemes,'' in \emph{International Conference on Availability,
  Reliability and Security}.\hskip 1em plus 0.5em minus 0.4em\relax IEEE, 2009,
  pp. 911--916.

\bibitem{zhan2010privacy}
J.~Zhan \emph{et~al.}, ``Privacy-preserving collaborative recommender
  systems,'' \emph{IEEE Transactions on Systems, Man, and Cybernetics, Part C
  (Applications and Reviews)}, vol.~40, no.~4, pp. 472--476, 2010.

\bibitem{shokri2009preserving}
R.~Shokri \emph{et~al.}, ``Preserving privacy in collaborative filtering
  through distributed aggregation of offline profiles,'' in \emph{Proceedings
  of the third ACM conference on Recommender systems}.\hskip 1em plus 0.5em
  minus 0.4em\relax ACM, 2009, pp. 157--164.

\bibitem{boutet2016privacy}
A.~Boutet \emph{et~al.}, ``Privacy-preserving distributed collaborative
  filtering,'' \emph{Computing}, vol.~98, no.~8, pp. 827--846, 2016.

\bibitem{Berkovsky07}
S.~Berkovsky \emph{et~al.}, ``Examining users' attitude towards privacy
  preserving collaborative filtering,'' in \emph{Data Mining for User Modeling
  On-line Proceedings of Workshop held at the}, 2007, p.~28.

\bibitem{Deshpande04}
M.~Deshpande and G.~Karypis, ``Item-based top-n recommendation algorithms,''
  \emph{ACM Transactions on Information Systems}, vol.~22, no.~1, pp. 143--177,
  2004.

\bibitem{Sarwar}
B.~Sarwar \emph{et~al.}, ``Item-based collaborative filtering recommendation
  algorithms,'' in \emph{Proceedings of the 10th international conference on
  World Wide Web}.\hskip 1em plus 0.5em minus 0.4em\relax ACM, 2001, pp.
  285--295.

\bibitem{tan2017efficient}
Z.~Tan and L.~He, ``An efficient similarity measure for user-based
  collaborative filtering recommender systems inspired by the physical
  resonance principle,'' \emph{IEEE Access}, 2017.

\bibitem{Das07}
A.~S. Das \emph{et~al.}, ``Google news personalization: scalable online
  collaborative filtering,'' in \emph{Proceedings of the 16th international
  conference on World Wide Web}.\hskip 1em plus 0.5em minus 0.4em\relax ACM,
  2007, pp. 271--280.

\bibitem{minhash}
A.~Z. Broder \emph{et~al.}, ``Min-wise independent permutations,'' in
  \emph{Proceedings of the thirtieth annual ACM symposium on Theory of
  computing}.\hskip 1em plus 0.5em minus 0.4em\relax ACM, 1998, pp. 327--336.

\bibitem{Chernoff}
H.~Chernoff, ``A measure of asymptotic efficiency for tests of a hypothesis
  based on the sum of observations,'' \emph{The Annals of Mathematical
  Statistics}, pp. 493--507, 1952.

\bibitem{Goldreich}
O.~Goldreich, ``Secure multi-party computation,'' \emph{Manuscript. Preliminary
  version}, pp. 86--97, 1998.

\bibitem{li2014item}
D.~Li \emph{et~al.}, ``Item-based top-n recommendation resilient to aggregated
  information revelation,'' \emph{Knowledge-Based Systems}, vol.~67, pp.
  290--304, 2014.

\end{thebibliography}
\end{document}